\documentclass[12pt]{article}
\pdfoutput=1
\usepackage[a4paper]{geometry}
\usepackage{jheppub, amsmath,amssymb,amsfonts,amsxtra, mathrsfs, makeidx,graphics,graphicx,amsthm,epsfig, bm,longtable,float, color,tikz,mathtools,xfrac,footnote,rotating, lscape, makecell, environ,mathtools, empheq}

\usepackage{subfig}

\usepackage{amsthm}
\usepackage{hyperref}

\usepackage{amsmath}
\usepackage{amsfonts}
\usepackage{commath}
\usepackage[english]{babel}
\usepackage{setspace}

\usepackage{epsfig,amssymb} 
\usepackage{amsmath}

\usepackage{amscd,color}
\usepackage{amsmath,graphicx}
\usepackage{graphicx}
\usepackage{verbatim}
\usepackage{latexsym}
\usepackage{amsmath,amsfonts,amssymb,amsthm}
\usepackage{amsmath,amsthm}

\def\be{\begin{equation}}
\def\ee{\end{equation}}

\newtheorem{theorem}{Theorem}

\title{
\begin{center}
On
the  Nelson-Seiberg theorem:\\
generalizations 
and counter-examples
\end{center}}

\author[a]{Antonio Amariti}
\author[b]{\!\!, Dario Sauro} 
\affiliation[a]{INFN, Sezione di Milano, Via Celoria 16, I-20133 Milano, Italy}
\affiliation[b]{Dipartimento di Fisica, Universit\`a di Milano, \\  Via Celoria 16, I-20133 Milano, Italy}

\emailAdd{antonio.amariti@mi.infn.it,dario.sauro@studenti.unimi.it}

\abstract{
In this note we elaborate on a recent counter-example to the Nelson-Seiberg theorem and to its generalizations. 
We provide sufficient conditions for the existence of such counter-examples, finding new ones. 
We claim that these counter-examples are connected with the presence of pairs of fields with opposite R-charge getting a VEV.
Furthermore we comment on the ”non-genericity” of our models, showing that none of them violates the original Nelson-Seiberg theorem.}

\begin{document}

\maketitle

\section{Introduction}

The Nelson Seiberg theorem \cite{Nelson:1993nf}  is a useful result for model building.
This theorem states that a 4d $\mathcal{N}=1$ WZ model must have an R-symmetry to break SUSY and that a spontaneously broken R-symmetry is a sufficient condition for breaking SUSY. 
There are many extensions and generalizations of this result (see \cite{Intriligator:2006dd,Ray:2007wq,Ferretti:2007rq,Ferretti:2007ec,Intriligator:2007py,Shih:2007av,Intriligator:2008fe,
Sun:2008va,Giveon:2009bv,Amariti:2009kb,Komargodski:2009jf,Sun:2011fq,Azeyanagi:2012pc,Curtin:2012yu,Kang:2012fn,Liu:2014ida,Sun:2018hnk} for an incomplete 
list of results).

Among these extensions here we focus on 
the revised versions of the theorem, that have been proposed in  \cite{Sun:2011fq} and \cite{Kang:2012fn} 
motivated by the asymmetry in the necessary and sufficient condition in the original theorem.
The revised versions of the theorem state that SUSY is broken if and only if there are more superfields $\Phi$
with  $R[\Phi]= 2$  than superfields $\Psi$ with $R[\Psi]= 2$.

The various derivations of the theorem share the same milestone,
  the assumption of calculability and of genericity of the superpotential. 
These assumption are motivated as follows: WZ models are derived as low-energy limits of gauge theories, 
where the gauge dynamics has been integrated out. 
It is then natural to assume that 
all the possible renormalizable interactions,  compatible with the symmetries of the model,
are generated in the WZ models and that the couplings are generic without any fine-tuning.

Observe that actually in the revised version a non generic issue arises, i.e. the requirement 
of having more fields with  R-charge 2  than fields with R-charge 0.
In spite of these results, recently, a counter-example to the revised Nelson-Seiberg theorem has been obtained in \cite{Sun:2019bnd}.
This is a model with one field with R-charge 2 and no field with vanishing R-charge.
The model displays a moduli-space of degenerate supersymmetric vacua, while according to the revised versions of NS theorem it should break SUSY.

In this paper we elaborate on this result,  providing sufficient conditions to generate counter-examples. 
By a deep analysis of the original counter-example of \cite{Sun:2019bnd} we generalize the construction
adding more fields and interactions.
Looking at many examples we then formulate a general statement, claiming the existence of sufficient conditions
for generating such counter-examples.
We complete our analysis discussing the notion of genericity applied to our models.

The paper is organized as follows.
In section \ref{NS-section} we review both the original version the NS theorem \cite{Nelson:1993nf} and  its generalizations \cite{Sun:2011fq,Kang:2012fn}. 
In section \ref{Originalcounter-section} we discuss the counter-example found in \cite{Sun:2019bnd} and we find some other cases by adding new fields with different values of the R-charge. In section \ref{Newcounter-section} we obtain new types of counter-examples with mass terms in the superpotential. 
% Such mass terms, absent in the original model and in the generalizations discussed in section \ref{Originalcounter-section}, are appealing because they are ubiquitous in  WZ models.
 In section \ref{Conditions-section} we study a set of sufficient conditions for finding counter-examples to \cite{Sun:2011fq,Kang:2012fn}.
  In section \ref{Non-gener-section} we  
  focus on the requirement of genericity 
  of the superpotential. This requirement appears in each version of the NS theorem and its importance in the proof is crucial. 
 We discuss how the assumptions of  section \ref{Conditions-section} lead us,
 after a non-linear change of basis in field space,  to a non-generic form of the superpotential, 
 even in absence of fine-tuning of the couplings.
 In this way we show that on one side our models do not violate the original NS theorem but on the other
side they are real counter-examples of \cite{Sun:2011fq,Kang:2012fn}.
%  In section \ref{Non-gener-section} we show that the hypotheses of section \ref{Conditions-section} violate the assumptions of %in absence of fine-tuning of the couplings. 
 %Since R-symmetry is not spontaneously broken in our models these are not counter-examples to the original theorem. However they are sufficient to say that the modified version of the theorem is not valid, or at least that should include a more precise definition of genericity.

\section{The NS theorem and its generalizations}
\label{NS-section}

The NS  theorem is a useful tool in model building, because it provides necessary and sufficient conditions for SUSY breaking in 4d $\mathcal{N}=1$ WZ models.
The necessary and sufficient conditions were discussed separately  in the original version \cite{Nelson:1993nf},
while in \cite{Sun:2011fq,Kang:2012fn} the authors tried to improve matters formulating the statement with a more appealing \emph{if and only if}. 
In this section we report both versions of the theorem.

\subsection{The original version}

Here we review  the original version of the NS theorem and its proof.
\begin{theorem}[NS theorem]\label{NS}
	Let us consider a 4d $\mathcal{N}=1$ renormalizable WZ model  with a generic superpotential. Then the presence of an R-symmetry is a necessary condition to have SUSY vacua while having a spontaneously broken R-symmetry is a sufficient condition.
\end{theorem}
This asymmetry between the necessary and sufficient conditions is  what motivated the authors of \cite{Sun:2011fq} and \cite{Kang:2012fn} to improve the statement. 
We turn now to the proof of the above theorem.

\begin{proof}[Proof of NS theorem]
In a purely chiral model we have three possibilities: either the model does not posses any global continuous symmetry or the model possesses a continuous global non-R symmetry or it possesses a continuous R-symmetry. We will now show that, for a generic superpotential, the first two cases lead to the existence of solutions of the SUSY vacuum equations, while in the last case SUSY is generically broken.
	
In the first case where no continuous global symmetry is present the SUSY vacuum equations for a chiral model constructed out of $n$ fields are $n$ equations in $n$ unknowns. Hence for a generic superpotential we shall be able to solve them, yielding SUSY vacua.
	
In the second case the holomorphy of the superpotential tells us that if there are $l$ generators of global symmetries, then the superpotential is a holomorphic function of $n-l$ variables. For $l=1$ and with $Q(\phi_n)=q_n \neq0$ these variables are $X_i=\frac{\phi_i}{\phi_n^{q_i/q_n}}$. Then the SUSY vacuum equations are $n-l$ equations in $n-l$ unknowns so, as in the previous case, SUSY is generically a symmetry of the vacuum.
	
In the third and last case the superpotential is a charged object, specifically it has R-charge 2. We suppose that $\phi_n$ gets a non-vanishing expectation value on the absolute minimum of the potential and denote its R-charge by $q_n$. We may then define the uncharged fields $\chi_i=\phi_i/\phi_n^{q_i/q_n}$. Thus the superpotential can be written as follows:
\begin{equation}
W = \phi_n^{2/q_n} f(\chi_i)
\end{equation}
The SUSY vacuum equations are then:
\begin{eqnarray}
f(\chi_i) = 0, \quad  \partial_j f(\chi_i)= 0
\end{eqnarray}
We have then $n$ equations in $n-1$ unknowns, so for a generic function $f(\chi_i)$ SUSY is broken at the global minimum of the potential. 
\end{proof}

\subsection{The revised version}
\label{revth}

We turn now to state and prove the modified version of the theorem which appeared in \cite{Sun:2011fq} and \cite{Kang:2012fn}. The focus of this revised version of the theorem is on the form of the most general 
superpotential as a sum of monomials of the fields. Here and in the following by general 
superpotential we mean that, for a given set of fields, every term compatible with R-symmetry is included and that there is no fine-tuning of the couplings.
We refer the reader to the statement given in \cite{Kang:2012fn}.

\begin{theorem}[The NS theorem revised]\label{NS_rev}
	In a Wess-Zumino model with a generic perturbative superpotential, SUSY is spontaneously broken at the global minimum if and only if the superpotential has an R-symmetry and there are more R-charge 2 fields than R-charge 0 fields for any consistent R-charge assignment \footnote{Observe that here we follow \cite{Sun:2019bnd} and include all possible terms compatible with R-symmetry. Hence we do not consider the possibility of having other global, non-R, $U(1)$-symmetries.}.
\end{theorem}

\begin{proof}
	The first thing to be done is a matter of notation, we will follow the conventions of \cite{Kang:2012fn}. We shall call $X_i$ all the R-charge 2 fields, $Y_j$ all the R-charge 0 fields and $A_k$ all the other fields. A generic perturbative superpotential can then be expanded as a sum of monomials in the following way. We use the summation convention on repeated indices  and sums are understood to be taken only on those subsets of the $A_k$ fields which yield an R-charge 2 superpotential.
	\begin{eqnarray}\label{generalW_revNS}
		W & = & X_i f^i(Y_j) + W_1\\
		W_1 & = & a^{ijk} X_i X_j A_k + b^{ijk} X_i A_j A_k + c^{ijk} Y_i A_j A_k+  d^{ijk} A_i A_j A_k + m^{ij} A_i A_j
		\nonumber
	\end{eqnarray}
If some of the $X_i$ or $A_k$ fields get a VEV we obtain SUSY breaking through the original version of the NS theorem. Thus the only way to get SUSY vacua is to let only the R-uncharged fields $Y_j$ get a VEV
\footnote{Here is a crucial observation: there is a further assumption in this step of the proof, i.e. no product of 
pairs of charged fields can get a VEV. In the following we will show that in the cases in which this product is uncharged under R-symmetry counter-examples are possible.}
. Setting $X_i=A_k=0$ is sufficient for setting to zero all the partial derivatives of $W_1$. We are then left with the $N_X$ equations $f^i(Y_j)=0$ in $N_Y$ unknowns. These have a solution for generic function $f^i$ if and only if $N_Y \geq N_X$, otherwise SUSY is broken, as was to be proved.
	
\end{proof}

\section{The first counter-example}\label{Originalcounter-section}

In this section we discuss  the counter-example presented in \cite{Sun:2019bnd}.
This is a counter-example to Theorem \ref{NS_rev} and it is obtained by considering 
one field with  R-charge 2 and no field with vanishing R-charge. 
Thus we should expect this model to display a SUSY breaking vacuum, yet we are able to solve the SUSY vacuum equations. The superpotential is:
\begin{equation}\label{W-original}
W = \mu^2 z + a z^2 \phi_1 + b z \phi_2 \phi_3 + c \phi_1^2 \phi_2
\end{equation}
The superpotential (\ref{W-original})  has  only one continuous internal symmetry
corresponding to  the R-symmetry. The absence of global, non-R $U(1)$-symmetries ensures the R-charge assignments are consistent. 
In \cite{Sun:2019bnd} the authors started from the four fields with the following R-charges and wrote down (\ref{W-original}) as the most general superpotential consistent with the R-charge assignments.
\begin{align}\label{R-charges-original}
R(z) & = +2; & R(\phi_1) & = -2; & R(\phi_2) & = +6; & R(\phi_3) & =-6
\end{align}
The reasoning can be reversed: we may start from (\ref{W-original}) and check that the R-charges are given precisely by (\ref{R-charges-original}). This freedom of reversing the starting point is due to the absence of non-R $U(1)$ global symmetries 
and to the absence of fine-tuning in the couplings. 
One then solves the SUSY vacuum equations:
there are SUSY vacua at $z=\phi_1=0$ and $\phi_2 \phi_3 = -\frac{a}{c}$. The model (\ref{W-original})  satisfies the hypotheses of the second theorem above, but  it displays SUSY vacua, so this is a counter-example. 

We shall now study the features of the superpotential (\ref{W-original}), because it is helpful to 
find a generalization. 
First, (\ref{W-original}) has only one dimensionful coupling, namely $\mu^2$. 
Another feature of the model is that the F-term equation of the field $z$ forces 
the product of two fields, with opposite R-charge, to get a VEV.
This will be a common property of the models  that we are going to discuss, and 
as a consequence we will denote with the letter $z$ the R-charge 
 2 fields and we refer to it as $z$-field, while other fields will be denoted with $\phi_i$.

The other interesting feature is that the two fields whose product gets a VEV, $\phi_2$ and $\phi_3$, appear only linearly in the superpotential.
Moreover, while the superpotential term $\phi_2 \phi_1^2$ is allowed, the coupling between $\phi_3$ and  $\phi_1^2$ is forbidden.
They  would either break R-symmetry or require the addition of another field. This remark will be important for finding a new counter-example. In this sense the field content of the model is made of one $z$-field  and three $\phi_i$ fields. Two of these  fields appear linearly in $W$ and their product gets a VEV.

\subsection{Generalizing the counter-example}

Another counter-example may be given by adding two more fields to the original superpotential (\ref{W-original}) 
with $R$-charges $R(\phi_4)=-R(\phi_5)=+14$.
The superpotential is 
\begin{equation}\label{new-counter-W}
W = \mu^2 z + a z^2 \phi_1 + b z \phi_2 \phi_3 + c \phi_1^2 \phi_2 + d \phi_3^2 \phi_4 + f z \phi_4 \phi_5
\end{equation}
By solving the F-terms we find SUSY vacua at $z=\phi_1=\phi_3=0$ and:
\begin{equation}
	\mu^2 + f \phi_4 \phi_5 = 0
\end{equation}
while $\phi_2$ is a modulus.  
The product which gets a VEV is now $\phi_4 \phi_5$ instead of $\phi_2 \phi_3$.
Both  $\phi_4$ and $\phi_5$ appear only linearly in the superpotential, while $\phi_3$ now couples 
quadratically to  $\phi_4$. 
We can generalize this construction to arbitrary higher values of the R-charges. 

Observe that we could have also considered letting $\phi_2$ appear quadratically 
in $W$ adding a pair of fields with R-charges $R(\phi_6)=-10$ and $R(\phi_7)=+10$ and superpotential 
\begin{align}\label{new-counter-W2}
W = \mu^2 z + & a z^2 \phi_1 + b z \phi_2 \phi_3 + c \phi_1^2 \phi_2 + d \phi_3^2 \phi_4
 + f z \phi_4 \phi_5 + g \phi_2^2 \phi_6 + h z \phi_6 \phi_7
\end{align}
Studying the F-terms one can see that the same analysis made for (\ref{new-counter-W}) applies to  
(\ref{new-counter-W2})  as well.

\section{A new counter-example}\label{Newcounter-section}

In section \ref{Originalcounter-section} we have discussed some generalizations
of the superpotential (\ref{W-original}),  studied in \cite{Sun:2019bnd}
as a prototypical counter-example of \cite{Sun:2011fq,Kang:2012fn}.
In this section we discuss a new type of counter-example, 
 by allowing the existence of another dimensionful coupling, namely a mass term.
 Here and in the rest of the paper we denote the massive fields with 
R-charge equal to one by the greek letter $\chi$.
We start considering five fields with R-charges
\begin{eqnarray}\label{massive_R_charges}
	R[z]= + 2, \quad R[\chi]=+1, \quad  R[\phi_1] = -1, \quad R[\phi_2]= +4, \quad  R[\phi_3]= -4
\end{eqnarray}
The most general superpotential consistent with this $U(1)_R$-symmetry is
\begin{equation}
\label{W-mass}
W = {\mu}^2 z + \frac{1}{2} m {\chi}^2 + a z \chi \phi_1 + b z \phi_2 \phi_3 + c {\phi_1}^2 \phi_2
\end{equation}
In analogy with the previous cases, we have no other continuous symmetry except the R-symmetry. 
There is a  discrete $\mathbb{Z}_2$ symmetry under which $\chi \rightarrow - \chi$ and $\phi_1 \rightarrow - \phi_1$. 
We find SUSY vacua for $z=\chi=\phi_1=0$ and
\begin{equation}\label{master_eq_mass}
\langle \phi_2 \phi_3 \rangle = - \frac{{\mu}^2}{b}
\end{equation}
This equation
\footnote{See \cite{Carpenter:2008wi} for a similar equation, relating the
 mechanism of runaways in  F-term SUSY breaking and the complexification of $U(1)_R$.}
represents a condition for the product of two fields with opposite R-charges, as in the examples discussed in \ref{Originalcounter-section}.

\subsection{Generalization}

In analogy with subsection \ref{new-counter-W} we can generalize the superpotential (\ref{W-mass}) by adding other fields.
We use the following recipe: we add pairs of new fields with opposite R-charges.  
We couple each new pair of fields in the superpotential such that their R-charges are  fixed. 
These new $W$ terms  must involve one of  the fields that appeared in the equation (\ref{master_eq_mass}), either  
$\phi_2$ or $\phi_3$.
It is crucial that this field appears quadratically in the final $W$, because it allows us to set it to zero by solving  the F-terms. 
Then the equation (\ref{master_eq_mass}) does not hold anymore.
At this step we can add (at least) two new fields, that play the same role played before by $\phi_2$ and $\phi_3$.
This is done by coupling either $\phi_2^2$ or $\phi_3^2$, with a field having the correct R-charge to yield an R-symmetric superpotential term.
 For example, if we consider $\phi_2^2$, then we must add to the theory a field $\phi_4$ with R-charge -6.
In addition we must add another field, $\phi_5$, such that $R[\phi_4]=-R[\phi_5]$.
  The new pair of fields couples naturally to $z$ with a term $d z \phi_4 \phi_5$.  We have then constructed the following superpotential:
\begin{align}
	W =  \mu^2 z & + \frac{1}{2} m \chi^2 + a z \chi \phi_1 + b z \phi_2 \phi_3
	 + c \phi_1^2 \phi_2 + d z \phi_4 \phi_5 + f \phi_2^2 \phi_4 
\end{align}
When we solve for the SUSY vacua  $\phi_2$ is set to zero. Furthermore also $z$, $\chi$ and $\phi_1$ are set to zero and we are left with 
the equation
\begin{equation}
\label{master_eq_mass_gen}
 \mu^2 + d \phi_4 \phi_5 = 0
 \end{equation}
that plays the role of (\ref{master_eq_mass}). 
 In this case we have a flat direction, parameterized by $\langle \phi_3 \rangle$.
 
We could also have started from  the superpotential (\ref{W-mass}) 
considering a coupling involving the field $\phi_3$ quadratically. 
This results in the addition of a new pair of fields of R-charges $\pm10$. 

If we consider both the pairs with R-charges $\pm 6$ and $\pm 10$ we observe that
a further mixing term is allowed by the R-symmetry.
We add this term to the superpotential because of the genericity assumption.
The field with R-charge $+6$ appears quadratically in the superpotential and the generalization of the equations 
 (\ref{master_eq_mass}) and  (\ref{master_eq_mass_gen}) involves only the fields with R-charge $\pm 10$.

The  key property that emerges from this analysis is that it is necessary to have 
 at least two fields with opposite R-charges and that
 these fields only appear linearly in the superpotential.
 The genericity assumptions than forces the coupling of these fields
 with the  R-charge 2 $z$-field. The F-term equation for the $z$-field implies the existence of SUSY vacua.
 %
 %
 %
 %%%%%%%%%%%%%%%%%%%%%%%%%%%%
 \section{Conditions for violating the revised-NS theorem}
 \label{Conditions-section}
 %%%%%%%%%%%%%%%%%%%%%%%%%%%%
 %
 %
 %
In this section we discuss the sufficient conditions for the existence of  counter-examples of \cite{Sun:2011fq,Kang:2012fn}. 
We first state the general result, then we motivate it through an example and eventually we  prove the result.

\begin{theorem}[Existence of counter-examples to \cite{Sun:2011fq,Kang:2012fn}]
\label{exis-theorem}
Consider the most general perturbative R-symmetric superpotential  $W$ 
with one field $z$ with R-charge $2$ and without any field with $R$ charge $0$. 
Then the existence of (at least) a pair of fields of opposite R-charge 
 both appearing only linearly in $W$ and not involved in any mass term
 is a sufficient condition for having SUSY vacua. 
On the SUSY vacua the product of such a pair of fields gets a VEV.
\end{theorem}

An important assumption in this statement is the existence of a pair of fields of opposite R-charges 
that are not involved in any mass term. Let us see in an example what happens if this 
assumption is violated. We consider the superpotential
\footnote{In this case we label the $\phi$ fields with their R-charge, while $R[z]=2$ and $R[\chi]=1$.} 
\begin{equation}\label{trouble-W}
W = \mu^2 z + \frac{1}{2} m \chi^2 + a z \chi \phi_{-1} + b z \phi_{3} \phi_{-3} + m_\phi \phi_{-1} \phi_{3}
\end{equation}	
We see that for $m_\phi \neq 0$ the above superpotential breaks SUSY at the tree-level. Both the fields of the pair $\phi_3$
and $\phi_{-3}$ appear only linearly in $W$ and they have opposite R-charges. 
SUSY-breaking is due to the mss term involving $\phi_3$,  i.e. $m_\phi \phi_{-1} \phi_3$.

\begin{proof}
We denote the field with R-charge $2$ by $z$, the fields with R-charge $1$ by $\chi_i$ and all the remaining fields with other R-charges by $\phi_i$. 
We consider only one field with $R$ charge $2$ for reasons of naturalness, but  our results would not change if we 
considered other fields with $R$-charge $2$. 
The sums over the $\phi_i$ fields are taken only on those subsets of fields consistent with R-symmetry. 
The superpotential we are considering is thus constrained to be:
		
	\begin{equation}\label{general_W}
		\begin{split}
		W = \mu^2 z & + \frac{1}{2} a^{ij} z \phi_i \phi_j  + \frac{1}{2} m^{ij} \chi_i \chi_j + b^{ij} z \chi_i \phi_j + \frac{1}{2} c^i z^2 \phi_i  \\ 
		& + \frac{1}{3!} d^{ijk} \phi_i \phi_j \phi_k + \frac{1}{2} f^{ijk} \phi_i \phi_j \chi_k + \frac{1}{2} \lambda^{ij} \phi_i \phi_j
		\end{split}
	\end{equation}
Note that the superpotential (\ref{general_W})  is consistent with (\ref{generalW_revNS}) in absence of
 fields with vanishing R-charge.  The SUSY vacuum equations are 	
	\begin{align}
		\partial_z W & = \mu^2 + \frac{1}{2} a^{ij} \phi_i \phi_j + b^{ij} \chi_i \phi_j + c^i z \phi_i = 0 \\
		\partial_{\chi_i} W & = m^{ij} \chi_j + b^{ij} z \phi_j + \frac{1}{2} f^{jki} \phi_j \phi_k = 0 \\
		\partial_{\phi_i} W & = a^{ij} z \phi_j + b^{ji} z \chi_j + \frac{1}{2} c^i z^2 + \frac{1}{2} d^{ijk} \phi_j \phi_k + f^{ijk} \phi_j \chi_k + \lambda^{ij} \phi_j = 0 \label{phi}
	\end{align}
	We consider, without loss of generality, the case in which the pair of fields mentioned in the hypotheses are $\phi_1$ and $\phi_2$. The fields $\phi_{1}$ and $\phi_{2}$  may belong to two multiplets of fields which are degenerate in the R-charge. Of course in such a case there is no \emph{preferred pair}. We label the indices of the (possibly absent) remaining fields with the same R-charges
of $\phi_1$ and $\phi_2$ $\alpha_1$, $\beta_1$, $\dots$ and $\alpha_2$, $\beta_2$, $\dots$ respectively. Then we label the indices of the remaining $\phi_i$ fields which do not satisfy the hypotheses of the theorem with upper hats $\hat{i}$, $\hat{j}$, $\dots$.	 		
	The requirements that the fields  $\phi_{1}$ and $\phi_{2}$  appear linearly in the superpotential 
	and that they have opposite R-charge  guarantee that in the $d^{1jk} \phi_1 \phi_i \phi_j$ and $d^{2jk} \phi_2 \phi_i \phi_j$ terms, where we are not summing over repeated indices, all the non-vanishing terms are those like $d^{1\hat{j}\hat{k}} \phi_1 \phi_{\hat{i}} \phi_{\hat{j}}$ and $d^{2\hat{j}\hat{k}} \phi_2 \phi_{\hat{i}} \phi_{\hat{j}}$. 
 The same reasoning holds for the $f^{1jk}\phi_1 \phi_i \chi_k$ terms as well.

 Eventually, the fields $\phi_{1}$ and $\phi_{2}$ are not coupled 
to the $\lambda^{ij}$ term. 
The F-term for $\phi_1$ then reads
	\begin{equation}
		\partial_1 W = a^{12} z \phi_2 + a^{1 \alpha_2} z \phi_{\alpha_2} + \frac{1}{2} d^{1\hat{j}\hat{k}} \phi_{\hat{j}} \phi_{\hat{k}} + f^{1\hat{j} k} \phi_{\hat{j}} \chi_k,
	\end{equation}
		%
	%with a totally analogous equation for $\phi_2$. Then setting:
	And a similar equation holds for $\phi_2$. 	
	Then setting
	\begin{align}
		z & = 0 & \chi_i & = 0 & \phi_{\widehat{i}} & = 0 & \phi_{\alpha_1} = \phi_{\alpha_2} = 0 
			\end{align}
these SUSY vacuum equations, as well as those for the $\chi_i$, $\phi_{\alpha_{1}}$, $\phi_{\alpha_{2}}$ 
 and $\phi_{\hat{j}}$ fields, are satisfied. Then we are left only with the $z$ F-term equation, which gives:
\begin{equation}
\label{master-eq}	
{\mu}^2 + a^{12} \phi_1 \phi_2 = 0
\end{equation}
representing the equation that turns on  the VEV $\langle \phi_1 \phi_2 \rangle$.		
\end{proof}

\section{On the non-genericity of the counter-examples}\label{Non-gener-section}

In this section we show that the models discussed here are counter-examples of the 
revised versions of the NS theorem, while they do not violate the original theorem.
The reason is that they are non-generic model in the sense of \cite{Nelson:1993nf}.
In order to see this non-genericity  we exploit a non-linear change of variables,  inspired to the one 
performed in the  NS theorem. Actually there is a difference in such a rescaling: while in the case of NS the author 
rescaled with respect of a field that gets a VEV, here we rescale with respect of a field that represents a
flat direction of the scalar potential, even if it this field is nowhere vanishing.
This allows us to show that the superpotentials studied here  are non-generic functions of the rescaled fields  
and that the origin of the rescaled fields  space is a SUSY vacuum.

%
%
%
%%%%%%%%%%%%%%%%%%%%%
\subsection{The original counter-example}
\label{nongen1}
%%%%%%%%%%%%%%%%%%%%%
%
%
%
We start our analysis by considering again the superpotential
(\ref{W-original}).
We consider a field that does not vanish on the vacuum, e.g. $\phi_2$.
We then redefine the  fields as
\begin{align}
	\varphi & = \frac{z}{\phi_2^{1/3}} & \varphi_1 & = \phi_1 \phi_2^{1/3} & \varphi_2 & = \phi_2^{1/3} & \varphi_3 & = \phi_2 \phi_3
\end{align}
The new fields are all un-charged with respect to the R-symmetry, except $\varphi_2$.
The superpotential (\ref{W-original}) becomes
\begin{equation}\label{new-china}
	W = \varphi_2 \left[ \mu^2 \varphi + a \varphi^2 \varphi_1 + b \varphi \varphi_3 + c \varphi_1^2 \right]
\end{equation}
The product $\phi_2 \phi_3$, that parameterizes a moduli-space of SUSY vacua, becomes $\varphi_3$
 in the new basis.  
 We expand this field as
   $\varphi_3=\Phi_0 + \delta \varphi_3$, where $\Phi_0 = - \frac{\mu^2}{b}$.
   In this way we can express the superpotential (\ref{new-china}) as
\begin{equation}\label{new-china2}
	W = \varphi_2 \left[ a \varphi^2 \varphi_1 + b \varphi \delta \varphi_3 + c \varphi_1^2 \right]
\end{equation}

The original argument of NS is that the superpotential can be re-written in the upper form $W=x f(y_i)$, where $x$ is the R-charged field getting a VEV and $y_i$ are all the other fields. The SUSY vacuum equations translate then in $n$ equations in $n-1$ unknowns, so SUSY is spontaneously broken for a generic function $f(y_i)$.

In the present case (\ref{new-china2}) the situation is a bit different. The crucial fact is that the function $f(y_i)$ is not a generic function  of the $y_i$'s, in fact it is a polynomial made up of monomials of degree 2 and 3
\footnote{Since the function $f(y_j)$ is a polynomial of degree 3 in its variables genericity means that it is a sum of monomials of any degree from 0 to 3. However having degree 0 monomials is possible only if the $x$ field is equal to the $z$ field in the original superpotential. Hence genericity requires linear terms in the $y_j$ variables to appear in the polynomial $f(y_j)$.}.
Here the $f(y_j)$ polynomial has no linear term because the only linear term in the 
un-rescaled superpotential, $\mu^2 z$, is canceled by the VEV of $\Phi_0$. 
 Such a function is vanishing at the origin, where it has vanishing gradient as well. It follows that we \emph{can} solve the SUSY equation.

%
%
%
%%%%%%%%%%%%%%%%%%%
\subsection{The new counter-example}
%%%%%%%%%%%%%%%%%%%
%
%
%
Here we extend the results of subsection \ref{nongen1} 
to the case of the \emph{massive} 
superpotential (\ref{W-mass}). 
We factor out $\phi_2$ as the non-vanishing field and re-define fields as
\begin{align}
	\varphi_1 & = \frac{z}{\phi_2^{1/2}} & \varphi_2 & = \frac{\chi}{\phi_2^{1/4}} & \varphi_3 & = \phi_1 \phi_2^{1/4} & X & = \phi_2^{1/2} & Y & = \phi_2 \phi_3
\end{align}
Writing $Y=Y_0 + \delta Y$. with $Y_0 = - \frac{\mu^2}{b}$, the superpotential takes the  form
\begin{equation}
	W = X \left( \frac{1}{2} \varphi_2^2 + a \varphi_1 \varphi_2 \varphi_3 + b \varphi_1 \delta Y + c \varphi_3^2 \right)
\end{equation}
Again we see that we have a polynomial in the R-uncharged rescaled fields made of monomials of degree 2 and 3, so our argument applies as above.

\subsection{The general case}

Here we generalize the result to the superpotential (\ref{general_W}).
We suppose that the pair of fields whose product gets a VEV are $\phi_1$ and $\phi_2$, and assume that they have R-charges $R(\phi_1)=-R(\phi_2)=q$. We define the rescaled fields as follows
\begin{align}
	X & = \phi_1^{2/q} & Y & = \phi_1 \phi_2 & Z & = \frac{z}{\phi_1^{2/q}} & X_i & = \frac{\phi_i}{\phi_1^{R(i)/q}} & K_i & = \frac{\chi_i}{\phi_1^{1/q}}
\end{align}
 where the indices of the fields $\phi_i$ run on $i \neq 1,2$.
  
We need to factor out the $X$ field to express the superpotential in the NS form.
We distinguish two possibilities for each term in $W$: either the $\phi_1$ field does not appear in such a 
term or it does. Let's consider three fields denoted as $\psi$, $\xi$, $\rho$ and a coupling constant $c$ and a general term in the superpotential $W  \supset c \psi \xi \rho$.
We define the R-charges of these fields as $R(\psi)  = p$,  $R(\xi)  = r$ and   $R(\rho)  = s $.
Then we rescale the fields as we have done before, where we 
denote  the rescaled field with the corresponding capitol greek letter.
We express the relations between the original fields and the rescaled ones in terms of the original fields for later convenience:
\begin{align}
	\psi & = \Psi \phi_1^{p/q} & \xi & = \Xi \phi_1^{r/q} & \rho & = P \phi_1^{s/q}
\end{align}
Hence the product of the three fields takes the form:
\begin{equation}
	\psi \xi \rho = (\phi_1)^{\frac{p+r+s}{q}} \Psi \Xi P  = X \Psi \Xi P,
\end{equation}
where in the last passage we exploited the fact that $p+r+s=2$. To consider the second case we set, for example, $\rho=\phi_1$, this time we find:
\begin{equation}
	\psi \xi \phi_1 = \Psi \Xi (\phi_1)^{\frac{p+r}{q}} \phi_1 = X \Psi \Xi
\end{equation}
The superpotential (\ref{general_W}) written in terms of $X$ and of the other rescaled fields
becomes
\begin{equation}
	\begin{split}
		W = & X \biggl\{ \mu^2 Z + a^{12} Z Y + a^{1i} Z X_i + a^{2i} Z Y X_i + \frac{1}{2} a^{ij} Z X_i X_j + \frac{1}{2} m^{ij} K_i K_j\\
		& + b^{ij} Z X_i K_j + \frac{1}{2} c^i Z^2 X_i + \frac{1}{2} d^{1ij} X_i X_j + \frac{1}{2} d^{2ij} Y X_i X_j + \frac{1}{3!} d^{ijk} X_i X_j X_k\\
		& + f^{1ij} X_j K_j + f^{2ij} Y X_i K_j + \frac{1}{2} f^{ijk} X_i X_j K_k + \frac{1}{2} \lambda^{ij} X_i X_j \biggr\}
	\end{split}
\end{equation}
Expanding as usual $Y=Y_0+\delta Y$ and using $\mu^2 + a^{12} Y_0 = 0$ we find that the function of the R-uncharged rescaled fields is a sum of monomials of degree two and three, as anticipated. 
The two requirements of the Theorem \ref{exis-theorem} force the non-genericity of the superpotential.

This is the reason for the breakdown of the arguments of the revised NS theorem.
Indeed throughout our calculations we have always kept the same hypotheses of \cite{Sun:2011fq} and \cite{Kang:2012fn}, i.e. of having a generic  superpotential with no fine-tuned couplings. Nevertheless we have just shown that the hypothesis of having more fields with R-charge two than fields with R-charge zero does not always imply tree-level SUSY breaking. This has been shown by using the same change of variables in field space which is exploited in the proof of the original NS theorem \cite{Nelson:1993nf} and by re-writing the superpotential in the form $W=x f(y_i)$. 
Under the hypotheses of the Theorem \ref{exis-theorem} the function $f(y_i)$ has no-linear 
terms and this feature is responsible for the presence of SUSY vacua, even if the models are generic 
in the sense of \cite{Sun:2011fq} and \cite{Kang:2012fn}.

\subsection{Comments on genericity}

In this final subsection we study the consequences of  relaxing some of the assumptions in Theorem \ref{exis-theorem}.
We focus on requirements of linearity and on the absence of mass terms for the fields which compose the pair with opposite R-charge.
We assume that at least one field is non-vanishing on the vacuum because of an equation like (\ref{master-eq}). 
This boils down to start solving the F-term for the field $z$ before the others.

We start by considering the case in which one of the two fields does not appear only linearly in the superpotential, 
but has a quadratic term as well, and we choose this field to be non-vanishing on the vacuum. 
We label such a field by $\phi_1$ and we assume that it has $R(\phi_1)=q$. 
The new superpotential term involving $\phi_1^2$ reads:
\begin{equation}
	W \supset g^i \phi_1^2 \phi_i,
\end{equation}
where the sum is understood to run on the fields having $R(\phi_i)=2 - 2q$. In the NS form this term becomes:
\begin{equation}
	g^i \phi_1^2 \phi_i = \phi_1^{2/q} g^i \phi_1^{2-2/q} \phi_i = \phi_1^{2/q} g^i X_i = X g^i X_i
\end{equation}
It follows that the polynomial now has a linear term: it becomes generic in the sense of NS 
and SUSY is broken at the origin, as expected from the Theorem \ref{NS}.

Similarly we can show that a linear term for the polynomial can be obtained 
by  adding a mass term for $\phi_1$.  The mass terms is $W \supset \lambda^{1j} \phi_1 \phi_j $
or in terms of the  rescaled fields $W \supset X \lambda^{1j} X_j$, and
it follows that we have a linear term.

Summarizing we have studied the connections between the various
assumptions and the requirement of genericity.
We conclude that all the models encountered so far 
are not counter-examples for the original NS theorem. 
This is connected, after field redefinition with respect of a nowhere vanishing field,
to the presence of an R-symmetry on the global minimum.

\section{Conclusions}

In this paper we considered  the original NS theorem \cite{Nelson:1993nf} and its revised versions \cite{Sun:2011fq,Kang:2012fn}. Inspired by  the counter-example to  \cite{Kang:2012fn},
recently obtained in \cite{Sun:2019bnd}, we found sufficient conditions for the breakdown of the arguments of the revised theorem. We carefully showed that this class of counter-examples does not invalidate the original NS theorem.

We also commented on the requirement of genericity, appearing in both the derivations of the NS theorem.
The sufficient conditions we have just mentioned fix
the structure of  our models after a non-linear change of variables, similar to the one  
carried out in the proof of the original theorem \cite{Nelson:1993nf}.

Even if here we assumed a notion of genericity based on the R-charges
our models fit in the assumptions of the second version of the NS theorem \cite{Kang:2012fn}. 
As observed above the models are counter-examples of these revised version because of the presence of 
products of pairs of fields with opposite R-charges getting a VEV.
This suggests us that forbidding this possibility should be added 
to the assumptions of \cite{Sun:2011fq,Kang:2012fn}.
 However, as already remarked, none of the models considered in 
 this paper are counter-examples of the original NS theorem.
 Indeed in the rescaled variables we have seen that the R-symmetry is unbroken in the vacuum.

We analyzed the 3d $\mathcal{N}=2$ case as well, even though we have not included such a discussion here. In this setting we can state and prove the counterparts of the two versions of the NS theorem. The same counter-examples of the second theorem arise in 3d as well. The analysis is complicated by the inclusion in the superpotential of quartic couplings.

As a conclusive remark we want to comment on the possibility of having supersymmetry breaking metastable states in the models discussed here. The authors of \cite{Sun:2019bnd} found that their model has saddle points only, we found the same result for our new counter-example (\ref{W-mass}). We would like to mention a further general result concerning SUSY-breaking metastable vacua. It is possible to show that in every model without R-charge -2 fields no SUSY-breaking metastable vacuum can be found near the origin of field space. 
 
\section*{Acknowledgments}

This work has been supported
in part by Italian Ministero dell'Istruzione, Universit\`a e Ricerca (MIUR),
in part by Istituto
Nazionale di Fisica Nucleare (INFN) through the ``Gauge Theories, Strings,
Supergravity" (GSS) research project and in part by
MIUR-PRIN contract 2017CC72MK-003.

\bibliographystyle{ytphys}
\bibliography{ref}

\end{document}